\documentclass[a4paper,english]{lipics-v2021}

\title{Interactive Exploration of the Temporal \texorpdfstring{$\boldsymbol{\alpha}$}{α}-Shape}

\author{Felix Weitbrecht}{Universit\"at Stuttgart, Germany}{weitbrecht@fmi.uni-stuttgart.de}{}{}
\authorrunning{F. Weitbrecht}
\Copyright{Felix Weitbrecht}


\usepackage{color}
\usepackage{amsmath}
\usepackage{graphicx}
\usepackage{amssymb}
\usepackage{amsbsy}

\nolinenumbers
\hideLIPIcs

\begin{document}
	
	\maketitle

	\vspace{-.6em}
	\begin{abstract}
		Shape is a powerful tool to understand point sets. A formal notion of shape is given by $\alpha$-shapes, which generalize the convex hull and provide adjustable level of detail. Many real-world point sets have an inherent temporal property as natural processes often happen over time, like lightning strikes during thunderstorms or moving animal swarms. To explore such point sets, where each point is associated with one timestamp, interactive applications may utilize $\alpha$-shapes and allow the user to specify different time windows and $\alpha$-values. We show how to compute the temporal $\alpha$-shape~$\alpha_T$, a minimal description of all $\alpha$-shapes over all time windows, in output-sensitive linear time. We also give complexity bounds on $|\alpha_T|$. We use $\alpha_T$ to interactively visualize $\alpha$-shapes of user-specified time windows without having to constantly compute requested $\alpha$-shapes. Experimental results suggest that our approach outperforms an existing approach by a factor of at least $\sim$52 and that the description we compute has reasonable size in practice. The basis for our algorithm is an existing algorithm which computes all Delaunay triangles over all time windows using $\mathcal{O}(1)$ time per triangle. Our approach generalizes to higher dimensions with the same runtime for fixed $d$.\vspace{.3em}\setlength{\parskip}{-.6em}
	\end{abstract}

	\vspace{-.5em}
	\begin{figure}[h]
		\centering
		\includegraphics[width=.265\columnwidth]{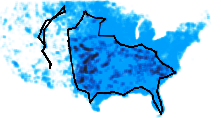}%
		\includegraphics[width=.265\columnwidth]{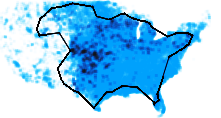}%
		\includegraphics[width=.265\columnwidth]{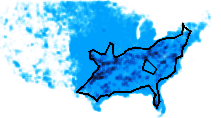}
		\includegraphics[width=.197\columnwidth]{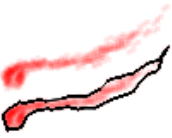}%
		\caption{\textbf{Left:} Storm events in the United States. The blue background is the entire point set, and for three time windows the points and $\alpha$-shapes (with different granularity) are overlaid in black. The $\alpha$-shapes are much smaller than the point sets. Left to right: 971 points vs.\ 80 edges, 917 points vs.\ 53 edges, and 761 points vs.\ 77 edges. \textbf{Right:} A particle swarm and an $\alpha$-shape representing its shape.}\label{fig:swarmAndStorm}
	\end{figure}

	\vspace{-.2em}
	\textbf{Corrigendum.} The algorithm output in this paper is not minimal as claimed. This error is corrected in~\cite{weitbrecht2026delaunay} with the algorithm in Section~5.2 and the discussion in Section~5.3.

	\vspace{-.2em}
	\section{Introduction}
		\vspace{-.3em}
		Visual analysis is a fundamental aspect of human-conducted examination of point data, in particular when there is a temporal aspect involved. We first take a closer look at two specific usecases to highlight the versatility of shape-based visualization.
		\begin{itemize}
			\item Visual analysis of storm events over time: Given a large collection of coordinate pairs which represent storm events, each observed or measured at one specific point in time, environmental scientists may be interested in aggregating storm events within various time windows to see how weather and climate patterns affect the frequency and spread of storm events on a larger time scale. This usecase is explored in~\cite{Haunert19}. For such purposes it is not necessary to display individual data points, instead the shape of points in given time windows provides a low complexity representation of the data which is both easier to parse for humans and more efficient to handle for computers, as seen in Figure~\ref{fig:swarmAndStorm}, left.
			\item Visual analysis of swarm movements: Many animal species exhibit swarm behavior, where animals travel as a large group in which each individual animal follows almost the same movement pattern. We may be given a large collection of 2D or 3D positions of animals moving as a swarm, with each data point being the position of some (unidentified) animal observed at one specific point in time. Zoologists or animal behavior scientists could use such data to study group dynamics or movement patterns. To gain an overview of the data it must be aggregated, for example using its shape, as in Figure~\ref{fig:swarmAndStorm}, right. By looking at the shape of data points within specific time windows, we could observe how the shape of the swarm changes as it advances, or how the swarm moves through particular areas of interest. An alternative approach is given by MotionRugs~\cite{8440823}, which visualize the individual velocities of a group of entities over time to provide a high-level overview of group dynamics.
		\end{itemize}

		A popular tool to examine point data is the Delaunay triangulation and its subcomplexes. They can be used to reconstruct geometric objects based on scan points sampled over time~\cite{aganj2007spatio,sussmuth2008reconstructing}, and for many other applications~\cite{edelsbrunner2011alpha} such as pattern recognition~\cite{vauhkonen2009identification}. In~particular, they can be used to represent the shape of a set of points, for example using facets of the convex hull. This approach can be generalized using $\alpha$-shapes~\cite{EKS83,edelsbrunner1994three}, which also allow concavities, disjointness and holes, making them a versatile tool for many usecases~\cite{edelsbrunner2011alpha}. While an edge is a facet of the convex hull iff the line through it induces a halfspace empty of points from the input, an edge is part of the $\alpha$-shape iff an empty ball of radius $\alpha$ passes through both vertices of that edge. This definition allows increasing the granularity of the resulting representation by choosing smaller values of $\alpha$. Real-world point sets often inherently have a temporal ordering to them, so it is also interesting to consider not just $\alpha$-shapes of the whole data set, but also $\alpha$-shapes of time windows within the data set. This naturally suggests the creation of interactive visualization applications which allow users to specify and adjust time windows and $\alpha$-values in order to explore a spatio-temporal data set.

		In our model, each data point is associated with exactly one timestamp, and timestamps cannot have multiple points associated with them. More complex scenarios, such as one animal/object being sampled multiple times, or multiple points appearing at the same time, can be represented using additional timestamps. One then only needs to adapt time windows to the modified timestamps. Our experiments suggest that the computational overhead due to these additional timestamps is tolerable in practice.

		\vspace{-.65em}
		\subparagraph{Related work.}
			Every $\alpha$-shape is a subcomplex of the Delaunay triangulation, so an easy way to compute $\alpha$-shapes is to first compute the Delaunay triangulation, and then pick out those edges which admit an empty $\alpha$-ball. The time required to compute the Delaunay triangulation makes this approach unsuitable for interactive applications, so it makes sense to precompute $\alpha$-shapes and query them based on user input. An approach like this is presented in~\cite{Haunert19}, where storm events in the United States, sampled between 1991 and 2000, are visualized using $\alpha$-shapes of different time windows and $\alpha$-values. All $\alpha$-shapes over all time windows are precomputed for a fixed value or range of $\alpha$. This is done by first computing for every point $p$ the set of potential neighbors $CPN(p)$, i.e.\ those points $q$ which are close enough that a ball of radius $\alpha$ could pass through $p$ and $q$, and then identifying those pairs $(p, q)$ which admit an \emph{empty} ball of radius $\alpha$. For larger values of $\alpha$, which lead to a less detailed representation of the shape, $CPN(p)$ grows larger, and most of the investigated pairs $(p, q)$ do not admit an empty $\alpha$-ball. Even for smaller $\alpha$-values many of the investigated pairs don't admit an empty $\alpha$-ball.
			
			When considering the whole Delaunay triangulation, a more efficient output-sensitive approach is given in~\cite{funke2020efficiently}, which shows how to compute the set $T$ of all Delaunay triangles over all time windows in time $\mathcal{O}(|T| \log n)$ using differential incremental constructions. This was later improved~\cite{weitbrecht2022linear} to $\mathcal{O}(|T|)$.

		\subparagraph{Contribution.}
			To avoid explicitly investigating pairs of points which don't admit an empty $\alpha$-ball, we make use of the Delaunay triangulation in an approach that also works in higher dimensions. Section~\ref{sec:alpha} studies how the temporal $\alpha$-shape, a minimal representation of all $\alpha$-shapes over all time windows and over all values of $\alpha$, is related to the Delaunay triangles of all time windows, and gives some upper bounds on its complexity. In Section~\ref{sec:algAlpha} we show how the temporal $\alpha$-shape can be computed based on the temporal Delaunay enumeration algorithm of \cite{funke2020efficiently,weitbrecht2022linear}. An implementation of both algorithms which works in arbitrary dimensions is available on GitHub~\cite{weitbrecht2023github}. Using this implementation for the preprocessing phase, we create an interactive visualization application for the temporal $\alpha$-shape in Section~\ref{sec:demo}. Experimental results in Section~\ref{sec:results} show that our approach compares favorably against an existing approach and suggest that precomputing the entire temporal $\alpha$-shape is practical for real-world applications. Section~\ref{sec:outlook} concludes with an outlook on future work.
			
			Our algorithm generalizes to arbitrary fixed dimension $d \geq 1$, but for ease of presentation we describe it in 2D.

	\section{Preliminaries}\label{sec:preliminaries}
		We are given a point sequence $P=\{p_1, p_2, \dots, p_n\} \subset \mathbb{R}^2$ in general position. Point indices give the temporal order and a point $p_i$ exists only at time $i$. The Delaunay triangulation~$DT(P)$ is that (unique) subdivision of the convex hull of $P$ which consists only of triangles whose open circumcircles contain no points from $P$ in their interior. $P_{i,j}$ is the time window $\{p_i, p_{i+1}, \dots, p_j\}$ and $T_{i,j}$ its Delaunay triangulation $DT(P_{i,j})$. If $e$ is an edge of a triangle~$t$, we call $t$ a \emph{coface} of $e$. $T$ is the set of all Delaunay triangles occurring over all time windows~$P_{i,j}$. With some abuse of notation: $T = \bigcup_{i<j} T_{i,j}$.

		An open disk of radius $\alpha > 0$ is an \emph{$\alpha$-ball}. If an $\alpha$-ball has both vertices of an edge $e$ on its boundary, we say it is an $\alpha$-ball of $e$. For a fixed value of $\alpha$, the $\alpha$-shape consists of exactly those edges $e \in \binom{P}{2}$ which have an $\alpha$-ball which contains no points from $P$.

		\subsection{Delaunay Triangulations and \texorpdfstring{$\boldsymbol{\alpha}$}{α}-shapes}\label{subsec:delaunayAlpha}
			By a sphere morphing argument, it can be seen that all $\alpha$-edges are also Delaunay edges. Additionally, every Delaunay edge $e$ is an $\alpha$-edge for a value range of $\alpha$ determined by $e$'s cofaces~\cite{EKS83}. Figure~\ref{fig:alphaDelaunay} shows the Delaunay triangulation and an $\alpha$-shape of a point set.

			\begin{figure}[h]
				\centering
				\includegraphics[width=.48\textwidth]{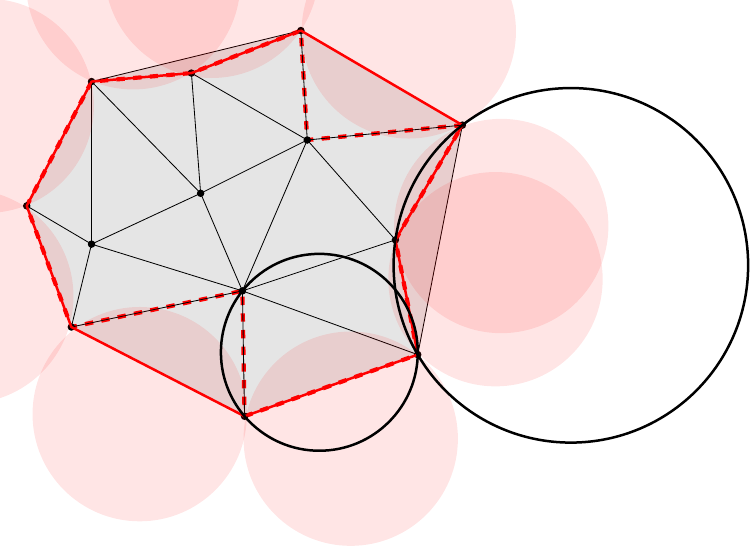}%
				\caption{A point set and its Delaunay triangulation in black, and an $\alpha$-shape shown with solid red edges. Some circumcircles of Delaunay triangles are shown with a black outline, and $\alpha$-balls are shown with a red overlay. As the $\alpha$-value becomes smaller, the resulting shape gains more detail over the convex hull. An even smaller $\alpha$-value might lead to the $\alpha$-shape shown with dashed red edges.}\label{fig:alphaDelaunay}
			\end{figure}
	
			One can consider the facets of the convex hull as degenerate Delaunay triangles. Their (degenerate) circumcircle is a halfspace and we define the corresponding circumradius to be $\infty$, and the circumcenter to be infinitely far away. This view will be convenient when deriving $\alpha$-edges from Delaunay triangulations, so we will use the terms \emph{triangle} and \emph{coface} to refer to triangles and facets alike. If $n=2$, we let the Delaunay triangulation be the edge between the two points.

		\subsection{Temporal Delaunay Triangle Enumeration}\label{subsec:delaunay}
			The \emph{hole triangulation} framework~\cite{funke2020efficiently} was introduced to compute the set $T$ of all Delaunay triangles over all time windows without asymptotic overhead for creation and deletion of triangles, taking overall time $\mathcal{O}(|T| \log n)$. This was later improved~\cite{weitbrecht2022linear} to $\mathcal{O}(|T|)$ using additional data structures which allow more efficient point location procedures. In this section we give a simplified overview of the resulting framework.
			
			The idea is to execute an incremental construction (\emph{IC}) of the Delaunay triangulation of every suffix of the point sequence: one IC for $P_{1,n}$, one IC for $P_{2,n}$, and so on. This process encounters the Delaunay triangulation of every time window as an intermediate state of one of these ICs, so it finds every Delaunay triangle of $T$. To avoid the overhead of creating the same triangle multiple times in different ICs, all ICs but the one for $P_{1,n}$ use a differential data structure: so-called \emph{hole triangulations} maintain only the difference between the ICs of two successive suffixes. The difference between the ICs of $P_{i-1,n}$ and $P_{i,n}$ is due to the removal of point $p_{i-1}$, so all Delaunay triangles which are not incident to $p_{i-1}$ in the IC of $P_{i-1,n}$ will also appear in the IC of $P_{i,n}$. The hole triangulation which is deployed for \linebreak the IC of $P_{i,n}$ thus only needs to triangulate the area covered by triangles incident to $p_{i-1}$ in the IC of $P_{i-1,n}$, as seen in Figure~\ref{fig:dataStructures}.
		
			\begin{figure}[h]
				\centering
				\includegraphics[width=.99\textwidth]{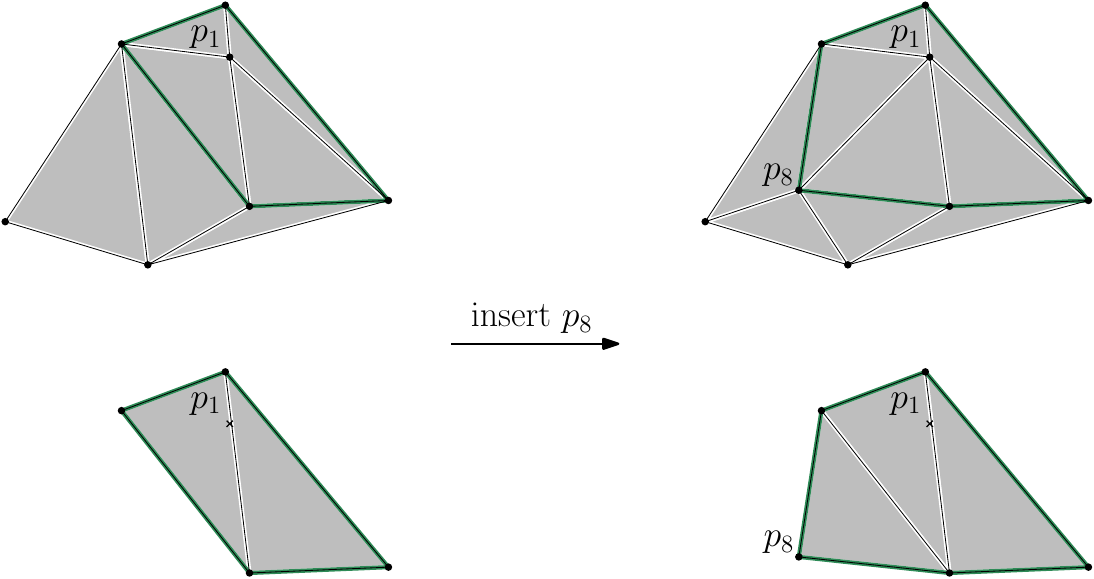}
				\caption{
					Data structures in the temporal Delaunay enumeration algorithm being updated as point~$p_8$ is inserted. \textbf{Top:} The full IC of $P_{1,n}$, $IC_{1,7}$ and $IC_{1,8}$. \textbf{Bottom:} The hole triangulation~$H_{2,7}$ being updated to $H_{2,8}$. The position of $p_1$ is shown only for reference. The link edges of $p_1$ are highlighted in green. Inserting $p_8$ in the full IC creates some new triangles, some of them incident to $p_1$. The link edges of $p_1$, which are the boundary of the hole triangulation of $P_{2,n}$, change, and along the changed link edges in $H_{2,8}$ new triangles are created which do not exist in $T_{1,8}$. Changes outside the link edges do not affect the hole triangulation.
				}\label{fig:dataStructures}
			\end{figure}
		
			We denote the intermediate states of the full incremental construction by $IC_{1,j}$ and those of hole triangulations by $H_{i,j}$. The first index is the start index of the respective time window and the second index represents the most recently inserted point. For visual aid, we arrange the intermediate states of these ICs into a matrix, using one row per IC:
			\begin{equation}
				\begin{matrix}
				\nonumber
					IC_{1,2} \rightarrow & IC_{1,3} \rightarrow & \dots \rightarrow & \dots \rightarrow & \dots \rightarrow & IC_{1,j} \rightarrow & \dots \rightarrow & \dots \rightarrow & IC_{1,n} \\
					& H_{2,3} \rightarrow & \dots \rightarrow & \dots \rightarrow & \dots \rightarrow & H_{2,j} \rightarrow & \dots \rightarrow & \dots \rightarrow & H_{2,n} \\
					& & & & & \vdots \\
					& & & H_{i,i+1} \rightarrow & \dots \rightarrow & H_{i,j} \rightarrow & \dots \rightarrow & \dots \rightarrow & H_{i,n} \\
					& & & & & \vdots \\
					& & & & & & & H_{n-2,n-1} \rightarrow & H_{n-2,n} \\
					& & & & & & & & H_{n-1,n}
				\end{matrix}
			\end{equation}
			This matrix is computed column by column, going top to bottom within each column. The hole triangulation $H_{i,j}$ answers the question ``what happens if we remove $p_{i-1}$ from $T_{i-1,j}$?'', so one could in theory reconstruct any $T_{i,j}$ by starting with $IC_{1,j}$ and successively applying the differences given by $H_{2,j}, H_{3,j}, \dots, H_{i,j}$, but this is not necessary in order to compute $T$.

			When a point is inserted into a Delaunay triangulation, it causes only local changes: triangles which contain that point in their circumcircle are destroyed, and they are replaced by the new triangles. For this reason, subsequent intermediate states of hole triangulations are often identical. In order to only compute updates to hole triangulations which actually cause structural changes we can use information from the rows above: The triangles in the hole triangulation $H_{i,j}$ all contain $p_{i-1}$ in their circumcircle because $H_{i,j}$ triangulates the hole created by the removal of $p_{i-1}$. So when two subsequent intermediate hole triangulation states $H_{i,j-1}$ and $H_{i,j}$ are not identical, there must be some new triangle $t$ appearing in~$H_{i,j}$. This triangle contains $p_{i-1}$ in its circumcircle, which means $p_{i-1}p_j$ is a Delaunay edge in $T_{i-1,j}$. While we may not have an explicit representation of $T_{i-1,j}$, all triangles and edges of $T_{i-1,j}$ are distributed across the data structures of $IC_{1,j}, H_{2,j}, H_{3,j}, \dots, H_{i-1,j}$. The computation order within each column is top-to-bottom, so if an update to $H_{i,j-1}$ is necessary, the edge~$p_{i-1}p_j$ must have been created in one of the updates ``above'' $H_{i,j}$, and we can use the creation of this edge to trigger an update from $H_{i,j-1}$ to $H_{i,j}$.
			
			The link edges of $p_{i-1}$ are the edges of its incident triangles which are themselves not incident to $p_{i-1}$. These link edges are also part of $H_{i,j-1}$, and they are the key to efficient point location in $H_{i,j-1}$. Simply put, once we know how the set of link edges changes, we can find the location of $p_j$ in $H_{i,j-1}$ by only traversing triangles which will be destroyed once~$p_j$ is inserted.
			
			The updates triggered by the creation of Delaunay edges form a directed acyclic graph (a \emph{DAG}) pointing downwards within each matrix column. The trick to accomplish point location in the full IC, i.e.\ the first row of the matrix, is to traverse this DAG in reverse order. Locating $p_j$ in $H_{j-2,j-1}$, where $p_j$ will definitely be inserted later, takes only $\mathcal{O}(1)$ time. Using pointers from link edges in hole triangulations to the triangles where they come from originally, we can find another hole triangulation (or the full IC) higher up in the matrix column where $p_j$ will also be inserted, by only visiting triangles which will be destroyed when $p_j$ is inserted later. This process is repeated until the first row of the matrix is reached, and its runtime is bounded by the time required to insert $p_j$ into the visited data structures, which is $\mathcal{O}(|T|)$ overall.
		
		\subsection{Rectangle Stabbing Queries}
			Given a set $R$ of axis-oriented boxes in $\mathbb{R}^d$, a rectangle stabbing query determines for a given query point $q \in \mathbb{R}^d$ the set of boxes $Q \subseteq R$ which contain $q$. Various flavors of this problem exist where the boxes are unbounded in one direction in all or some dimensions. Elaborate data structures exist which solve variants of this problem with good worst-case performance and linear space consumption, see for example~\cite{chan2022orthogonal} and related work therein.
			
			A more straightforward approach, which sacrifices worst-case complexity for practicality, is given by \emph{cs-box-trees}~\cite{Agarwal2002}: $d$-dimensional boxes are associated with $2d$-dimensional points by associating the 2 directions of each of the $d$ dimensions with 2 of the $2d$ dimensions. A~\emph{kd-tree} is built on these $2d$-dimensional points, and during its construction each inner node gets $2d$ \emph{priority leaf} nodes which store the boxes which have the most extreme values in both directions of the $d$ dimensions. Each node of the tree is associated with the bounding box of its children. Constructing a cs-box-tree is possible in $\mathcal{O}(n \log n)$ time, and query complexity is $\mathcal{O}(n^{1-1/d} + k \log n)$, where $k$ is the number of boxes reported in the query. Queries recursively visit child nodes with bounding boxes which contain the query point $q$ and report for all visited leaf nodes those boxes which contain $q$.

	\section{The Temporal \texorpdfstring{$\boldsymbol{\alpha}$}{α}-Shape}\label{sec:alpha}
		We want to compute a description of all $\alpha$-shapes over all time windows and over all $\alpha$-values, i.e.\ for each Delaunay edge occurring in $T$ we would like to know for which time windows and $\alpha$-values it is an $\alpha$-edge. Activity spaces, defined for fixed $\alpha$-value using only 2 temporal parameters in~\cite{Haunert19}, formalize this notion. We include $\alpha$ as a third parameter:
		
		\begin{definition}
			The set of triples $(i,j,\alpha) \subseteq \{1, 2, \dots, n\} \times \{1, 2, \dots, n\} \times \mathbb{R}^+$ which correspond to a time window $P_{i,j}$ and $\alpha$-value for which an edge $e$ is an $\alpha$-edge is called the $\alpha$-edge activity space $L^\alpha_e$.
		\end{definition}
		We will see in the following sections that these $L^\alpha_e$ can be represented efficiently in a compact manner by partitioning them into cuboids, each defined by 2 values for each of the 3~parameters. We can naturally define cardinality based on such partitions:
		
		\begin{definition}
			The cardinality $|L^\alpha_e|$ of an $\alpha$-edge activity space $L^\alpha_e$ is the minimum cardinality of a partition of $|L^\alpha_e|$ into cuboids.
		\end{definition}
		We can now define the temporal $\alpha$-shape and its cardinality:

		\begin{definition}
			The temporal $\alpha$-shape, denoted by $\alpha_T$, is the set of all $\alpha$-edge activity spaces.
		\end{definition}

		\begin{definition}
			The cardinality $|\alpha_T|$ of the temporal $\alpha$-shape $\alpha_T$ is the sum of the cardinalities of the $\alpha$-edge activity spaces contained in $\alpha_T$: $|\alpha_T| := \sum_{L^\alpha_e \in \alpha_T} |L^\alpha_e|$
		\end{definition}

		To get an intuition for activity spaces without the intricacies that come with $\alpha$-edges, let us first look at how they work on Delaunay triangles. Then, after also looking at activity spaces of Delaunay edges, we will take a closer look at $\alpha$-edge activity spaces. Note that even though $i$ and $j$ are discrete parameters, the figures in this paper will depict them as if they were continuous for ease of presentation.
		\begin{definition}
			For a Delaunay triangle $t$, the set of tuples $(i,j) \subseteq \{1, 2, \dots, n\} \times \{1, 2, \dots, n\}$ which correspond to time windows $P_{i,j}$ in which $t$ is a Delaunay triangle is called the Delaunay activity space $L^D_t$.
		\end{definition}
		
		A triangle $t \in T$ is a Delaunay triangle in a time window $P_{i,j}$ iff \emph{(a)} all vertices of $t$ have point indices between $i$ and $j$ and \emph{(b)} the open circumcircle $C$ of $t$ contains no points of $P_{i,j}$. Let $i_\mathrm{lower}$ and $i_\mathrm{upper}$ be the lowest and highest point indices among the vertices of $t$, then \emph{(a)}~is fulfilled iff $i \leq i_\mathrm{lower} \wedge i_\mathrm{upper} \leq j$.
		
		For \emph{(b)}, $C \cap P_{i_\mathrm{lower},i_\mathrm{upper}}$ is already empty because $t \in T$, so we consider the points inside~$C$ with a point index less than $i_\mathrm{lower}$, i.e.\ $C \cap P_{1,i_\mathrm{lower}-1}$. Let $i_\mathrm{before}$ be the highest point index among them, or $-\infty$ if no such points exist. Similarly, let $i_\mathrm{after}$ be the lowest point index in $C \cap P_{i_\mathrm{upper}+1,n}$, or $\infty$ if no such points exist. Now \emph{(b)} is fulfilled iff $i_\mathrm{before} < i \wedge j < i_\mathrm{after}$. These conditions for $i$ and $j$ are independent, so $L^D_t$ is a rectangle $]i_\mathrm{before},i_\mathrm{lower}] \times [i_\mathrm{upper},i_\mathrm{after}[$, as seen in Figure~\ref{fig:delaunayAS}. The first dimension gives the range of the lower end of time windows in which $t$ is a Delaunay triangle, and the second dimension that of the upper end. The Delaunay triangles of time window $P_{i,j}$ are exactly those whose activity space contains the point $(i,j)$.

		\begin{figure}[h]
			\centering
			\includegraphics[width=.8\textwidth]{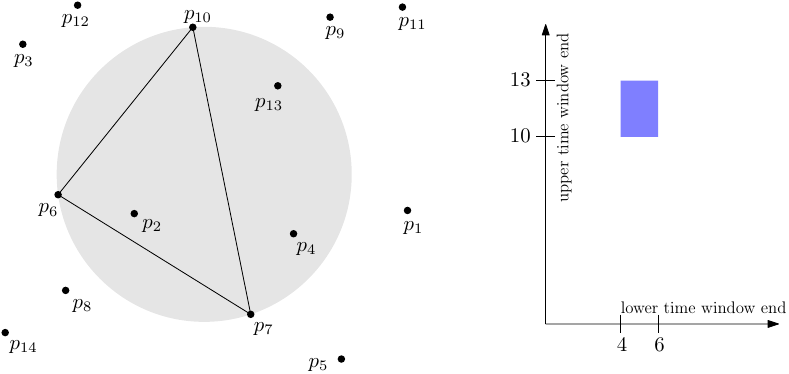}%
			\caption{\textbf{Left:} A Delaunay triangle $t \in T$ and some points of the whole point sequence. Some~of the points are inside the circumcircle of $t$. \textbf{Right:} $L^D_t$, the activity space of $t$, which is bounded by $i_\mathrm{lower}=6$, $i_\mathrm{upper}=10$, $i_\mathrm{before}=4$, and $i_\mathrm{after}=13$.}\label{fig:delaunayAS}
		\end{figure}

		All $\alpha$-edges are also Delaunay edges, and every Delaunay edge is an $\alpha$-edge for a certain range of $\alpha$~\cite{EKS83}, as indicated in Figure~\ref{fig:alphaDelaunay}. We explore this relationship using activity spaces in Sections~\ref{subsec:activitySpaceDelaunay} and~\ref{subsec:activitySpaceAlpha}. We then study the complexity of $\alpha_T$ in Section~\ref{sec:complexityAlpha} before Section~\ref{sec:algAlpha} describes how, given $T$, we can efficiently compute $\alpha_T$.

		\subsection{Activity Spaces of Delaunay Edges}\label{subsec:activitySpaceDelaunay}
			Activity spaces for Delaunay edges are defined analogously to those of Delaunay triangles:

			\begin{definition}
				For a Delaunay edge $e$, the set of tuples $(i,j) \subseteq \{1, 2, \dots, n\} \times \{1, 2, \dots, n\}$ which correspond to time windows $P_{i,j}$ in which $e$ is a Delaunay edge is called the Delaunay activity space $L^D_e$.
			\end{definition}

			\begin{figure}[h]
				\centering
				\includegraphics[width=\textwidth]{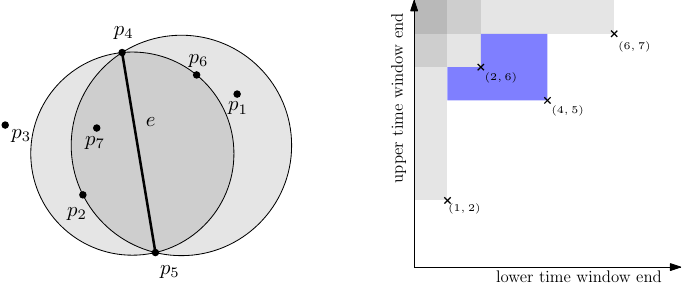}%
				\caption{\textbf{Left:} A Delaunay edge $e$ and some points. The pairs $\{p_1, p_2\}$, $\{p_2, p_6\}$, and $\{p_6, p_7\}$ prevent $e$ from being a Delaunay edge. More such pairs exist, but their influence on $L^D_e$ is dominated by that of the 3 given pairs. \textbf{Right:} Subtracting the rectangles corresponding to the 3 pairs (in transparent grey) from $]-\infty, 4] \times [5, \infty[$ yields the Delaunay activity space of $e$, $L^d_e$ (in blue).}\label{fig:delaunayFaceAS}
			\end{figure}

			For Delaunay triangles, the disk which needs to be empty is fixed, as it is simply the circumcircle. But for Delaunay edges there are infinitely many disks to consider. For this reason, Delaunay activity spaces of edges are not necessarily rectangular, instead they are shaped like a staircase as in Figure~\ref{fig:delaunayFaceAS}, right. Intuitively, the bottom right staircase corner~$(4, 5)$ represents the smallest time window enclosing both endpoints of the edge, and moving to the left or to the top from there means expanding the time window to include lower or higher indices. As the time window grows, points will appear closer and closer to the edge, so the circumcircles of cofaces of $e$ will be smaller and smaller. Eventually there may be points which are so close to $e$ that $e$ is no longer a Delaunay edge, like it is the case at $(2, 6)$. Such points create the steps of the staircase.
			
			More formally, Delaunay edge activity spaces have the staircase property:
			
			\begin{definition}\label{def:staircase}
				An activity space $L$ has the staircase property iff it can be obtained by taking a rectangle which is unbounded towards $-\infty$ in the dimension representing the lower time window end, and unbounded towards $\infty$ in the other dimension, and then subtracting from it a set of rectangles shaped in the same way. That is, we could take a rectangle $\left]-\infty,i_\mathrm{lower}\right] \times \left[i_\mathrm{upper},\infty\right[$ and subtract from it a set of rectangles like $\left]-\infty,i_1\right] \times \left[i_2,\infty\right[$.
			\end{definition}

			For a Delaunay edge $e$ occurring in $T$ to be a Delaunay edge in time window $P_{i,j}$, \emph{(a)}~applies like before, and for \emph{(b)}, we need to look at what forces every open disk passing through the vertices of $e$ to be non-empty. We can always find such an empty disk for convex hull edges, simply by choosing its center sufficiently far away. For all other edges, first observe that all disks which pass through all vertices of $e$ are centered on a common line. If we grow such a disk $S$ in one direction, i.e.\ move its center along this line, more and more points on that side of $e$ will be inside $S$ (and they will stay inside as we keep moving~$S$). If we now move $S$ the other way, eventually all these points will be outside of $S$ again, and the last point that was inside $S$, say $p_{i_\mathrm{last}}$, characterizes how far $S$ needs to be moved back to be empty of points on that side of $e$. If, at this moment, no points on the other side of $e$ are inside $S$, we have found a disk satisfying \emph{(b)}. Figure~\ref{fig:delaunayFaceAS}, left, shows some disks which are empty in some time windows.

			But if $S$ is not empty, no such disk can exist because moving $S$ in either direction only puts more points of that side inside $S$. In this case, we can choose any point $p_{i_\mathrm{other}}$ inside~$S$ to find a pair $\{p_{i_\mathrm{last}}, p_{i_\mathrm{other}}\}$ which certifies that $e$ is not a Delaunay edge, and we might find even more pairs by repeating the same procedure on the other side of $e$. Over all time windows, there may be many such pairs, and \emph{(b)} is fulfilled iff the time window $P_{i,j}$ contains no such pair. Subtracting the rectangles $]-\infty, i_\mathrm{last}] \times [i_\mathrm{other}, \infty[$ corresponding to these pairs (w.l.o.g.\ $i_\mathrm{last} < i_\mathrm{other}$) from $]-\infty, i_\mathrm{lower}] \times [i_\mathrm{upper}, \infty[$ thus yields the activity space $L^D_e$, which is staircase-shaped as seen in Figure~\ref{fig:delaunayFaceAS}.

		\subsection{Activity Spaces of \texorpdfstring{$\boldsymbol{\alpha}$}{α}-Edges}\label{subsec:activitySpaceAlpha}
			Note first that any Delaunay edge $e$ is always an $\alpha$-edge for certain values of $\alpha$, so if we project $L^\alpha_e$ down into the two temporal dimensions, it will be identical to $L^D_e$. But an $\alpha$-ball of $e$ may be centered on either side of $e$, and we will consider both sides separately. This gives us the option to work with oriented edges, and to determine which edges in the $\alpha$-shape admit two empty $\alpha$-balls, i.e.\ where the $\alpha$-shape has no area and instead is only one edge thick. We fix one side of $e$ as the front and define accordingly:
			
			\begin{definition}
				The $\alpha$-edge front activity space $L^{\alpha_\mathrm{front}}_e$ of an edge $e$ is that subset of $L^\alpha_e$ which corresponds to time windows and $\alpha$-values for which an empty $\alpha$-ball of $e$ is centered in front of $e$.
			\end{definition}
			The $\alpha$-edge activity space $L^\alpha_e$ then is the union of the front and back activity spaces.
			
			Consider the circumcircles $C_\mathrm{front}$ and $C_\mathrm{back}$ of the cofaces of $e$ in the Delaunay triangulation of some time window $P_{i,j}$, and their circumradii $r_\mathrm{front}$ and $r_\mathrm{back}$. By the empty circumcircle property of Delaunay triangles, an $\alpha$-ball of $e$ is empty iff it is centered between $C_\mathrm{front}$ and $C_\mathrm{back}$, see Figure~\ref{fig:alphaBallRange}. So the circumradii of $e$'s cofaces determine the $\alpha$-value range. If $C_\mathrm{front}$ is centered behind $e$, the $\alpha$-value range of the front is empty. Otherwise, the front circumradius upper bounds $\alpha$. The lower bound is given by the back circumradius if $C_\mathrm{back}$ is also centered in front of $e$, and otherwise by the radius of the smallest disk which has both vertices of $e$ on its boundary, say $r_\mathrm{min}$. Translating this condition into a value range for $\alpha$, we get that $r_\mathrm{back} \leq \alpha \leq r_\mathrm{front}$ in case $C_\mathrm{back}$ and $C_\mathrm{front}$ are centered in front of $e$, and $r_\mathrm{min} \leq \alpha \leq r_\mathrm{front}$ if only $C_\mathrm{front}$ is centered in front of $e$. If $C_\mathrm{front}$ is also centered behind $e$, no empty $\alpha$-ball may exist in front of $e$ and the value range is empty.

			\begin{figure}[h]
				\centering
				\includegraphics[width=.569\textwidth]{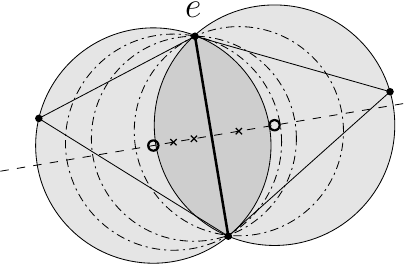}%
				\caption{The cofaces of Delaunay edge $e$ determining the possible centers of empty $\alpha$-balls of $e$. All disks through the vertices of $e$ are centered on the dashed line. The circumcircles of $e$'s cofaces are filled in transparent grey, and their centers are marked with circles on the dashed line. The crosses on the dashed line indicate possible centers of empty $\alpha$-balls of $e$, which are shown with a dash dotted outline.}\label{fig:alphaBallRange}
			\end{figure}

			If $e$ has the same cofaces in different time windows, the $\alpha$-value range is also the same, so we can characterize $L^{\alpha_\mathrm{front}}_e$ based on all cofaces of $e$ in $T$. For every coface pair $(t_\mathrm{front}, t_\mathrm{back})$, we get a (potentially empty) cuboid whose extent in the temporal dimensions is the intersection of $L^D_{t_\mathrm{front}}$ and $L^D_{t_\mathrm{back}}$. The $\alpha$-value range of this cuboid is determined as above. The resulting non-empty cuboids partition $L^{\alpha_\mathrm{front}}_e$. Unlike the partitions we compute in Section~\ref{sec:algAlpha}, this partition is not necessarily minimal.

			The requirement for the $\alpha$-ball to be centered in front of $e$ implies that, if we project $L^{\alpha_\mathrm{front}}_e$ down into the two temporal dimensions, it may be a proper subset of $L^D_e$. This is because any point $p_{i_\mathrm{close}}$ in front of $e$ which together with $e$ forms a triangle whose circumcenter is behind $e$ cuts the cuboid $\left]-\infty,i_\mathrm{close}\right] \times \left[i_\mathrm{close},\infty\right[ \times \left]-\infty,\infty\right[$ off from $L^{\alpha_\mathrm{front}}_e$. Still, the staircase property applies w.r.t.\ the temporal dimensions because such cutoffs only shorten the staircase in one of the temporal dimensions.

		\subsection{Complexity of the Temporal \texorpdfstring{$\boldsymbol{\alpha}$}{α}-Shape}\label{sec:complexityAlpha}
			We give some upper bounds on the complexity of $\alpha_T$. We first observe that $\mathcal{O}(|T|)$ does not necessarily bound $|\alpha_T|$, like in the following example. In the worst case, a single $\alpha$-edge activity space $L^\alpha_e$ may have to be partitioned into $\Omega(n^2)$ cuboids even if $|T| \in \mathcal{O}(n)$, as in Figure~\ref{fig:badAlphaActivitySpace}.
	
			\begin{figure}[h]
				\centering
				\includegraphics[width=.7\textwidth]{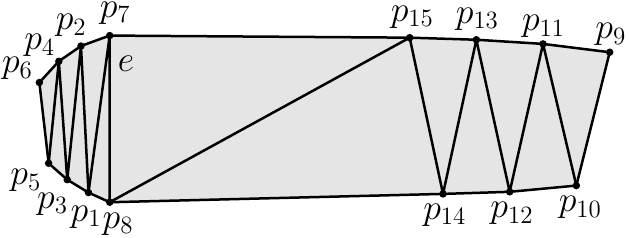}%
				\caption{A temporally ordered point set and its Delaunay triangulation where $|T| \in \Theta(n)$, but $\alpha_T \in \Omega(n^2) = \Omega(n \cdot |T|)$. The central edge $e = p_7p_8$ is a Delaunay edge in every time window that contains $p_7$ and $p_8$, so points $p_1, \dots, p_6$ can not form edges or triangles with $p_9, \dots, p_{15}$. Due to the correlation between the geometric and temporal ordering of the two groups of points on either side of $e$, only $\mathcal{O}(n)$ Delaunay triangles can exist over all time windows, hence $|T| \in \Theta(n)$. However, $|L^\alpha_e|$~has quadratic complexity for the following reason. Every time window $P_{i,j}$ with $i \leq 7$ has a different lower bound on the $\alpha$-value range of $e$'s $\alpha$-ball, because the closest point to $e$ on its left side comes closer and closer with decreasing $i$. Analogously, every time window $P_{i,j}$ with $j \geq 8$ has a different upper bound. Therefore every combination of $i \leq 7$ and $j \geq 8$ produces a non-empty cuboid with a distinct $\alpha$-value range. There are $\Omega(n^2)$ such combinations of $i$ and $j$, so $|L^\alpha_e| \in \Omega(n^2)$.}\label{fig:badAlphaActivitySpace}
			\end{figure}
			
			This example generalizes easily to arbitrary $n$ and $d$ (but not to arbitrary $|T|$), so we get the following lemma:
			\begin{lemma}\label{lem:alphaBlowup}
				There exist point sequences $\{p_1, p_2, \dots, p_n\} \subset \mathbb{R}^d$ with $|\alpha_T| \in \Omega(n \cdot |T|)$ for arbitrary $n$ and $d$.
			\end{lemma}
			Luckily we can get a better bound on $|\alpha_T|$ than the trivial $\mathcal{O}(n^2 \cdot |T|)$ with a more global view:

			\begin{theorem}\label{thm:alphaComplexity2}
				$|\alpha_T| \in \mathcal{O}(n \cdot |T|)$, and this bound is tight in the worst case.
			\end{theorem}
			\begin{proof}
				Consider a Delaunay edge $e$ and its $\alpha$-edge activity space $L^\alpha_e$. The $\alpha$-value range of~$e$ in a given time window is completely determined by its two Delaunay cofaces. So we can get a partition $\mathcal{P}_e$ of $L^\alpha_e$ into cuboids by grouping together time windows in which $e$ has the same cofaces. That is, for every two cofaces $t_1, t_2 \in T$ of $e$, $\mathcal{P}_e$ contains a cuboid whose extent in the temporal dimensions is the intersection of $L^D_{t_1}$ and $L^D_{t_2}$, unless that intersection is empty. $|L^\alpha_e|$ is the minimum cardinality of a partition of $L^\alpha_e$ into cuboids, so we have $|L^\alpha_e| \leq |\mathcal{P}_e|$ and we can bound $|\alpha_T| = \sum_{L^\alpha_e \in \alpha_T} |L^\alpha_e| \leq \sum_{L^\alpha_e \in \alpha_T} |\mathcal{P}_e|$.
	
				There exist at most $n$ cofaces of $e$ in $T$. Any one coface can contribute at most $n - 1$ cuboids to $\mathcal{P}_e$ (one potential cuboid with each of the other cofaces of $e$). The $|T|$ Delaunay triangles each have at most 3 edges\footnote{Facets of the convex hull, which we consider as degenerate Delaunay triangles, only have a single edge.}, so we can bound the summed partition cardinalities as follows: $\sum_{L^\alpha_e \in \alpha_T} |\mathcal{P}_e| \leq \sum_{t \in T} 3 \cdot (n - 1) \in \mathcal{O}(n \cdot |T|)$. By Lemma~\ref{lem:alphaBlowup}, this bound is tight in the worst case.
			\end{proof}
			
			To give some further context, $|T|$ can be bounded by $\mathcal{O}(n^{1+\lceil d/2 \rceil})$ in the general case, and in the expected case by $\mathcal{O}(n \cdot \mathrm{polylog}(n))$ if the points have certain properties~\cite{weitbrecht2023complexityeurocg}. With Theorem~\ref{thm:alphaComplexity2} we get for 2D that $|\alpha_T| \in \mathcal{O}(n^3)$, which may still be improved using a more global analysis. The bound of Theorem~\ref{thm:alphaComplexity2} may be tight in terms of $|T|$, but better bounds in terms of $n$ are likely. In particular, the data sets examined in Section~\ref{sec:results} produce $|\alpha_T| < 3 \cdot |T|$.

	\section{Computing the Temporal \texorpdfstring{$\boldsymbol{\alpha}$}{α}-Shape}\label{sec:algAlpha}
		The temporal $\alpha$-shape $\alpha_T$ is the set of $\alpha$-shape activity spaces $L^\alpha_e$ of every Delaunay edge occurring in $T$. Our approach to compute these activity spaces based on $T$, as computed by the Delaunay enumeration algorithm described in Section~\ref{subsec:delaunay}, is as follows.
		
		The Delaunay activity spaces $L^D_t$ of the front cofaces of $e$ partition the staircase-shaped Delaunay activity space $L^D_e$ into rectangles, and the same is true for the back cofaces. Figure~\ref{fig:rectLists} shows such front and back partitions on the left. We overlay these partitions to get a partition into smaller rectangles. Together with their $\alpha$-value ranges which can be computed based on the corresponding cofaces as described in Section~\ref{subsec:activitySpaceAlpha}, these rectangles produce the cuboids which partition the $\alpha$-edge activity space $L^\alpha_e$. We do this separately for both sides of edges. We fix one side of some Delaunay edge $e$ as the front side and describe how to compute a minimal representation of $L^{\alpha_\mathrm{front}}_e$. Repeating the process for the back side of $e$ and for all other edges produces $\alpha_T$.
	
		We begin by listing Delaunay triangles with their edges in a useful order in Section~\ref{subsec:algOrdering}. In all subsequent sections we consider each edge, and each side of each edge, separately. We compute the rectangular activity spaces $L^D_t$ of each triangle $t$ and order these rectangles into four lists in Section~\ref{subsec:algRelisting}. We then simplify these lists to omit rectangles which do not influence $L^{\alpha_\mathrm{front}}_e$ in Section~\ref{subsec:algCleaning}. After assigning some pointers between the lists in Section~\ref{subsec:algLinking}, we can intersect the activity spaces of the cofaces of $e$ to compute $L^{\alpha_\mathrm{front}}_e$ in Section~\ref{subsec:algIntersecting}. The final result will be a set of cuboids for both sides of every Delaunay edge $e$ occurring in $T$.
		
		After the first, global, step in Section~\ref{subsec:algOrdering}, our algorithm works on each edge independently, so the steps of Sections~\ref{subsec:algRelisting} through~\ref{subsec:algIntersecting} are easily parallelizable.

		\subsection{Ordering Delaunay Triangles}\label{subsec:algOrdering}
			Consider some Delaunay edge $e$ occurring in $T$ and fix one side of it as the front side. Remember that $L^D_e$ is staircase-shaped, and that activity spaces of Delaunay triangles are rectangular. Since $e$ has exactly one front coface in every time window it is a Delaunay edge in, this staircase can be partitioned into the rectangular activity spaces of the front cofaces of $e$ over all time windows. We would like a nicely ordered representation of this partition.
			
			For this purpose, we will slightly modify the Delaunay enumeration algorithm of~\cite{funke2020efficiently,weitbrecht2022linear} which effectively computes an IC (\emph{incremental construction}) of each suffix $P_{i,n}$ of $P$. Separate instances of the same Delaunay edge may be created in multiple ICs, and different cofaces are computed with each instance. We modify the algorithm as follows:

			\begin{itemize}
				\item We make every IC maintain the set of all Delaunay edges ever created in it.
				\item We make both sides of every instance $e^*$ of an edge $e$ maintain a list of cofaces that were created on that side with that instance: $cl^\mathrm{front}_{e^*}$ and $cl^\mathrm{back}_{e^*}$. These lists are ordered by time of coface creation, i.e.\ by the lower boundary in the second dimension of activity spaces.
			\end{itemize}

			We can now represent the partition of the front activity space of $e$ by listing the $cl^\mathrm{front}_{e^*}$, which we get from the ICs in which an instance of $e$ was created, in the order of the ICs. Figure~\ref{fig:columnLists} shows a few examples of such representations. We can compute this representation for both sides of all Delaunay edges in one pass by creating an empty list of coface lists for both sides of each edge, $cll^\mathrm{front}_e$ and $cll^\mathrm{back}_e$, and then iterating over the ICs. For each IC we append each computed Delaunay edge instance's coface lists to the corresponding lists of coface lists: $cl^\mathrm{front}_{e^*}$ to $cll^\mathrm{front}_e$ and $cl^\mathrm{back}_{e^*}$ to $cll^\mathrm{back}_e$. This requires runtime linear in the number of Delaunay triangles over all time windows.
			
			\begin{figure}[h]
				\centering
				\includegraphics[width=\textwidth]{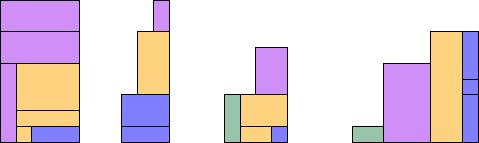}%
				\caption{The front activity space of four edges represented by their front coface lists $cll^\mathrm{front}_e$. For each edge, the rectangles corresponding to the same coface list $cl^\mathrm{front}_{e^*}$ are shown in the same color.}\label{fig:columnLists}
			\end{figure}

			Let $i_\mathrm{lower}$ and $i_\mathrm{upper}$ be the lowest and highest point indices of the vertices of $e$. Let us associate the terms \emph{left}, \emph{right}, \emph{top}, and \emph{bottom} with the temporal dimensions, matching the alignment in the figures so far. Each triangle $t$ is associated with a rectangular activity space~$L^D_t$, so we can also view the coface lists as rectangle lists, and we may refer to the $L^D_t$ as rectangles. We call a coface list \emph{grounded} if its lowest rectangle is aligned with the bottom staircase boundary, and \emph{floating} otherwise. We can now make some useful observations about the coface lists.

			\begin{lemma}\label{lem:obsDocked}
				Every rectangle is aligned with the right staircase boundary ($i_\mathrm{lower}$), or the bottom staircase boundary ($i_\mathrm{upper}$), or both.
			\end{lemma}
			\begin{proof}
				A point which forms a triangle with $e$ can have an index smaller than $i_\mathrm{lower}$, or larger than $i_\mathrm{upper}$, but not both. So at least one of the indices of $e$'s endpoints remains as $i_\mathrm{lower}$ or $i_\mathrm{upper}$ in every coface of $e$.
			\end{proof}

			\begin{lemma}\label{lem:obsShared}
				Two rectangles have the same left boundary iff they are in the same list. 
			\end{lemma}
			\begin{proof}
				The left boundary of $L^D_t$ is given by the index of the largest index point inside the circumcircle of $t$ whose index is smaller than any index of vertices of $t$. For triangles computed in the IC of $P_{1,n}$ no such point exists (and we use $-\infty$ as the index). For triangles computed in the IC of some $P_{i,n}$, that index is $i-1$ because ICs of suffixes of $P$ only compute triangles which do not appear in previous ICs.
			\end{proof}

			\begin{lemma}\label{lem:obsStacked}
				The rectangles within each list are stacked on top of each other, i.e.\ the top boundary of every rectangle aligns with the bottom boundary of the next rectangle (if one exists).
			\end{lemma}
			\begin{proof}
				When a point $p_{i_\mathrm{new}}$ inserted into a Delaunay triangulation destroys a coface $t$ of $e$ (and thus defines the top boundary of $L^D_t$ as $i_\mathrm{new}$), it either destroys $e$ entirely or it creates a new coface $t'$ of $e$. But then the bottom boundary of $L^D_{t'}$ must be $i_\mathrm{new}$ because $i_\mathrm{new}$ is the highest point index among the vertices of $t'$. 
			\end{proof}

			\begin{lemma}\label{lem:obsOnTop}
				All rectangles of grounded coface lists are entirely below all rectangles of floating coface lists.
			\end{lemma}
			\begin{proof}
				No rectangle can touch any rectangle of a floating coface list along the list's left boundary due to Lemmas~\ref{lem:obsDocked} and~\ref{lem:obsShared} and the fact that the staircase has no holes.
			\end{proof}

			In the following sections we focus only on $L^{\alpha_\mathrm{front}}_e$, the description for $L^{\alpha_\mathrm{back}}_e$ is analogous. We can then merge $L^{\alpha_\mathrm{front}}_e$ and $L^{\alpha_\mathrm{back}}_e$ to get $L^\alpha_e$, or keep them separate to preserve the ability to distinguish between the two sides of $e$.

		\subsection{Ordering Activity Space Rectangles into Lists}\label{subsec:algRelisting}
			We begin by computing the activity spaces $L^D_t$ of all Delaunay triangles $t$ we found as cofaces of some edge $e$ in Section~\ref{subsec:algOrdering}. By Lemma~\ref{lem:obsDocked}, every rectangle touches the staircase boundary either on the right, or the bottom, or both. For both sides of $e$, we create a \emph{bottom} list and a \emph{right} list. The \emph{bottom} list holds all rectangles touching the bottom boundary, ordered left to right. The \emph{right} list holds all remaining rectangles, ordered bottom to top. The bottom right rectangle is only listed in the \emph{bottom} list. Figure~\ref{fig:rectLists} shows an example of \emph{bottom} and \emph{right} lists on the left.

			Note that these lists are computed for both sides of $e$ during the computation of $L^{\alpha_\mathrm{front}}_e$, and again for both sides during the computation of $L^{\alpha_\mathrm{back}}_e$. This is because because Section~\ref{subsec:algCleaning} will modify these lists based on which side of $e$ is being considered. We describe how to compute these lists for the front side of $e$, the process for the back side is analogous.

			The \emph{bottom} list is constructed by iterating left to right over the coface lists $cl^\mathrm{front}_{e^*}$ which are grounded, always adding the bottom-most rectangle to the \emph{bottom} list.
			
			By Lemma~\ref{lem:obsOnTop}, we can fill the \emph{right} list bottom-up by first only considering grounded coface lists, and then only considering  floating coface lists. Because the staircase has no holes, if a grounded coface list, say $cl^\mathrm{front}_{e^*}$, contains a rectangle $L^D_t$ of the \emph{right} list, all grounded coface lists right of $cl^\mathrm{front}_{e^*}$ can only contain rectangles underneath $L^D_t$. We can thus start filling the \emph{right} list by iterating right to left over grounded coface lists, bottom up within each coface list, and adding all rectangles which do not touch the bottom boundary to the \emph{right} list.
			
			The coface lists within $cll^\mathrm{front}_e$ are ordered left to right according to their left boundary. Due to the staircase property and Lemma~\ref{lem:obsDocked}, this order is also a bottom to top order when restricted to floating coface lists. So to complete the \emph{right} list, we simply iterate left to right over the floating coface lists in $cll^\mathrm{front}_e$, bottom up within each coface list, and add all rectangles to the \emph{right} list.

		\subsection{Simplifying The Activity Space of a Delaunay Edge}\label{subsec:algCleaning}
			We noted in Section~\ref{subsec:activitySpaceAlpha} that $L^{\alpha_\mathrm{front}}_e$ can be partitioned into cuboids whose $\alpha$-value range depends on the cicumcircles of the cofaces of $e$ in the respective time windows. We would like to avoid explicitly computing cuboids with empty $\alpha$-value ranges. So we will now discard the front rectangles of $e$ which force an empty value range, i.e.\ those which correspond to Delaunay triangles whose circumcenter lies behind $e$. We modify the \emph{bottom} and \emph{right} lists of the front as follows.
			
			Following the logic of Lemma~\ref{lem:obsStacked}, inserting points into a Delaunay triangulation which admits an empty $\alpha$-ball in front of $e$ will move the circumcenter of the front coface of $e$, which limits the radius of that $\alpha$-ball, closer and closer to $e$, until it eventually is behind~$e$. In the temporal dimensions of activity spaces, inserting points corresponds to decreasing the value in the first dimension or increasing the value in the second dimension. Removing the undesirable rectangles therefore cuts off some rectangles of the \emph{bottom} list from the left, and some rectangles of the \emph{right} list from the top. Figure~\ref{fig:rectLists} shows this step on the top left.
			
			\begin{figure*}[t]
				\centering
				\includegraphics[width=\textwidth]{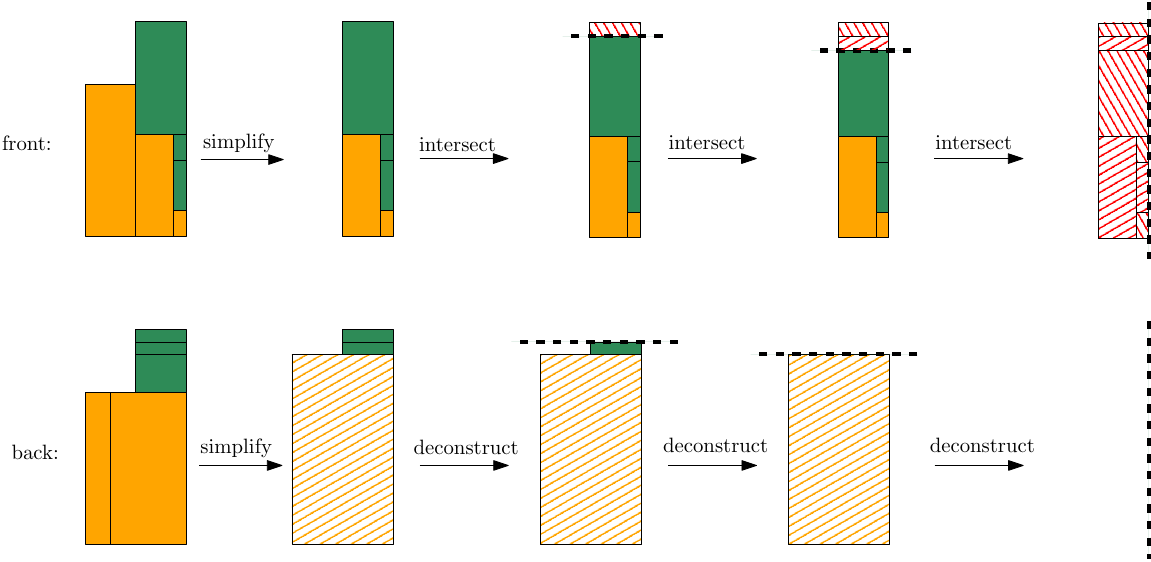}%
				\caption{The \emph{bottom} and \emph{right} lists of $e$'s front and back as they are simplified and intersected. On the left, the \emph{bottom} lists are shown with orange rectangles and the \emph{right} lists with green rectangles. In the simplification step we remove some front rectangles from the left, and merge some back rectangles from the bottom right. The merged rectangle $L^\mathrm{dummy}$, shown in an orange hatch pattern, is larger than the rectangles it represents. Then the back staircase is deconstructed rectangle by rectangle, always intersecting the removed rectangle with the front rectangles. The intersection step corresponds to cutting off all remaining rectangles from the top if the back rectangle is from the \emph{right} list, or from the left if it is from the \emph{bottom} list. The resulting intersected rectangles are shown in a red hatch pattern where the front rectangles used to be. Together with their value range of $\alpha$, the resulting intersected rectangles yield the cuboids which partition $L^{\alpha_\mathrm{front}}_e$.}\label{fig:rectLists}
			\end{figure*}

			We would also like to avoid creating multiple cuboids with the same $\alpha$-value range, guaranteeing that the set of cuboids we compute for $L^{\alpha_\mathrm{front}}_e$ is minimal. To that end, we will combine all back rectangles of $e$ which have no effect on the front $\alpha$-value range. Using similar logic as before, as points are inserted, the circumcenter of back cofaces of $e$ may eventually be in front of $e$. But until then, the back cofaces do not influence the value range of empty $\alpha$-balls in front of $e$. We will merge all rectangles belonging to such cofaces into one combined rectangle, i.e.\ modify the \emph{bottom} and \emph{right} lists of the back.

			Starting from the bottom right, we traverse the \emph{bottom} list to the left, and the \emph{right} list to the top, to find out for how long in each temporal dimension $e$'s back coface has its circumcenter behind $e$. If either direction extends further than the bottom right rectangle, say~$L^D_t$, we replace $L^D_t$ with a dummy rectangle, say $L^\mathrm{dummy}$, whose extent in the two temporal dimensions covers the time windows in which $e$'s back coface has its circumcenter behind $e$. We also remove any rectangles which are now covered by the enlarged bottom right rectangle $L^\mathrm{dummy}$. Figure~\ref{fig:rectLists} shows this step on the bottom left.
			
			Observe that no rectangles can be covered only partially: By definition, the left boundary of $L^\mathrm{dummy}$, say $i_\mathrm{left}$, splits all time windows into two groups; those time windows that begin at or before $i_\mathrm{left}$, in which any back coface of $e$ must have its circumcenter in front of $e$, and those time windows that begin after $i_\mathrm{left}$, in which any back coface of $e$ must have its circumcenter behind $e$. No back coface of $e$ can exist in time windows of both groups, so no rectangle can cross the left boundary of $L^\mathrm{dummy}$. Analogous logic applies in the other temporal dimension.

			$L^\mathrm{dummy}$ does not correspond to exactly one Delaunay triangle, and in fact we may actually have extended the staircase to include time windows in which $e$ is not even a Delaunay edge. But this is not a problem. In Section~\ref{subsec:algIntersecting} we intersect the back rectangles with the front rectangles, which we did not extend. Intersecting with them guarantees that we do not include additional time windows in the computed cuboids. Most importantly, Lemma~\ref{lem:obsDocked} and the staircase property still hold with $L^\mathrm{dummy}$.
			
		\subsection{Linking Activity Space Rectangles}\label{subsec:algLinking}
			To conveniently traverse staircases, we would like to know for each rectangle $L^D_{t_b}$ in the \emph{bottom} list which rectangle of the \emph{right} list, if any, is above $L^D_{t_b}$. Similarly, we also need to know for each rectangle $L^D_{t_r}$ in the \emph{right} list which rectangle of the \emph{bottom} list, if any, is left of $L^D_{t_r}$. Due to Lemma~\ref{lem:obsDocked} at most one such rectangle can exist for each $L^D_{t_b}$ or $L^D_{t_r}$. Since the staircase has no holes, the rectangles of the \emph{bottom} list and those of the \emph{right} list fit together like two elaborate puzzle pieces, and walking along their shared seam by traversing both lists in parallel, starting at the bottom right, allows us to find all neighbor relations in one pass.

		\subsection{Computing the Activity Space of an \texorpdfstring{$\boldsymbol{\alpha}$}{α}-Edge}\label{subsec:algIntersecting}
			Let us give a short summary of how $L^{\alpha_\mathrm{front}}_e$ relates to the Delaunay activity spaces of $e$'s cofaces, as described in Section~\ref{subsec:activitySpaceAlpha}. The $\alpha$-value range of $L^{\alpha_\mathrm{front}}_e$ in time window $P_{i,j}$ depends on the cofaces $t_\mathrm{front}$ and $t_\mathrm{back}$ of $e$ in $T_{i,j}$. If $t_\mathrm{front}$ has its circumcenter behind $e$, the $\alpha$-value range is empty, otherwise its circumradius is the upper bound of the $\alpha$-value range. If $t_\mathrm{back}$ has its circumcenter in front of $e$, its circumradius is the lower bound of the $\alpha$-value range, otherwise $r_\mathrm{min}$, the radius of the smallest disk which passes through the vertices of $e$, is the lower bound. Every pair $(t_\mathrm{front}, t_\mathrm{back})$ over all time windows yields a cuboid whose extent in the temporal dimensions is the intersection of $L^D_{t_\mathrm{front}}$ and $L^D_{t_\mathrm{back}}$, and whose $\alpha$-value range is determined as above. These cuboids partition $L^{\alpha_\mathrm{front}}_e$. Having simplified the activity space rectangles belonging to the front cofaces and back cofaces in Section~\ref{subsec:algCleaning}, we are guaranteed that every non-empty intersection of front and back rectangles yields a distinct, non-empty, $\alpha$-value range, i.e.\ the computed set of cuboids representing $L^{\alpha_\mathrm{front}}_e$ is minimal. The following approach computes these cuboids in output-sensitive linear time.

			Figure~\ref{fig:rectLists} illustrates our approach. We deconstruct the back staircase one rectangle at a time, always ensuring that the remaining rectangles still fulfill the staircase property. To that end, we always take away either the left-most rectangle of the \emph{bottom} list, or the top-most rectangle of the \emph{right} list. Removing a rectangle with a neighbor in the other list would violate the staircase property, so we use the neighbor relations computed in Section~\ref{subsec:algLinking} to determine whether a rectangle still has a neighbor in the other list. Note that by Lemma~\ref{lem:obsDocked}, there is always at least one rectangle we can take away without violating the staircase property. Every back rectangle $L^D_{t_\mathrm{back}}$ we take away is intersected with the front rectangles, and every non-empty intersection yields a cuboid as described above. By remembering our progress along the front's \emph{bottom} and \emph{right} lists, we functionally deconstruct the front staircase as well: front rectangles which are covered by $L^D_{t_\mathrm{back}}$ are discarded, front rectangles which are only intersected by $L^D_{t_\mathrm{back}}$ are cut off along the boundaries of $L^D_{t_\mathrm{back}}$.
			
			Observe that after modifying the front rectangles and back rectangles in Section~\ref{subsec:algCleaning}, the area covered by the back staircase is a superset of the area covered by the front staircase, so preserving the staircase property on the back means we automatically preserve the staircase property on the front. Any front rectangle which we cut off along $L^D_{t_\mathrm{back}}$'s boundaries is functionally split into two rectangles, one which remains in the front staircase, and one which is used for the creation of a cuboid.
			
			We store our current intersection position along the \emph{bottom} and \emph{right} lists on the front and the back. Using the neighbor relations computed in Section~\ref{subsec:algLinking} we can easily find all front rectangles which intersect $L^D_{t_\mathrm{back}}$ in linear time. In the end, we will have computed the cuboids which represent $L^{\alpha_\mathrm{front}}_e$.
			
			After repeating this process for the other side of $e$, and for all other Delaunay edges, we have achieved our desired result and get the following theorem.
			
			\begin{theorem}
				There exists an algorithm to compute the temporal $\alpha$-shape $\alpha_T$, which is a description of all $\alpha$-shapes over all time windows and all values of $\alpha$, for a set of timestamped points in $\mathbb{R}^d$ in output-sensitive linear time $\mathcal{O}(|\alpha_T|)$ for arbitrary fixed $d$.
			\end{theorem}
			\begin{proof}
				The Delaunay algorithm~\cite{funke2020efficiently,weitbrecht2022linear} runs in output-sensitive linear time $\mathcal{O}(|T|)$. Lemma~4 of~\cite{funke2020efficiently}, which states that the number of Delaunay simplices and the number of Delaunay faces are of the same order, generalizes to arbitrary $d$ trivially. So by the fact that every Delaunay face is an $\alpha$-face, the runtime of the Delaunay algorithm is also $\mathcal{O}(|\alpha_T|)$. Ordering simplices into lists with faces in Section~\ref{subsec:algOrdering} takes time $\mathcal{O}(|T|)$. Sections~\ref{subsec:algOrdering} to~\ref{subsec:algLinking} only perform linear sweeps, so their runtime is $\mathcal{O}(|T|)$ as well. Intersecting rectangles for face $f$ in Section~\ref{subsec:algIntersecting} requires time linear in the number of cuboids computed for $f$. Since the number of cuboids we compute is minimal, and the $\alpha$-face activity space of either side of $f$ can not be described with asymptotically less space than a cuboid partition, this final step requires $\mathcal{O}(|\alpha_T|)$ time over all Delaunay faces.
			\end{proof}

			Computing cuboids separately for $L^{\alpha_\mathrm{front}}_e$ and $L^{\alpha_\mathrm{back}}_e$ may create overlapping cuboids for $e$ in time windows in which both sides of $e$ admit an empty $\alpha$-ball. But the simplification step of the back staircase partition merges all these time windows into one combined rectangle, so the cuboids created by the intersection step for the potentially problematic time windows correspond exactly to the rectangles of the front staircase partition. Over all edges, these are at most $\mathcal{O}(|T|)$, so the asymptotic complexity of the computed representation is unaffected.
			
	\begin{figure*}[!h]
		\centering
		\includegraphics[width=.3333\textwidth]{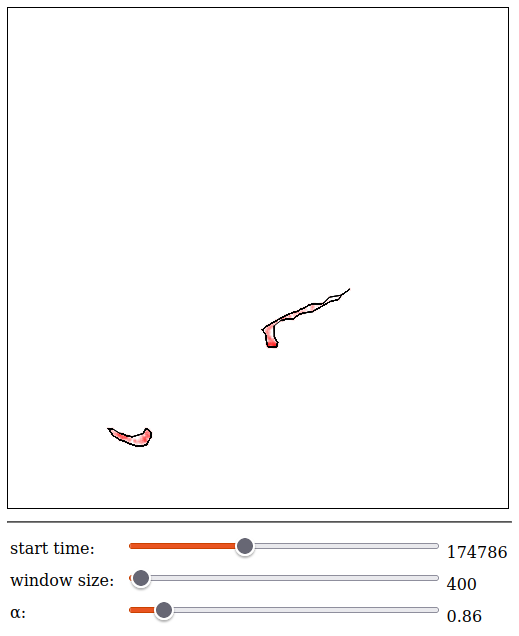}%
		\includegraphics[width=.3333\textwidth]{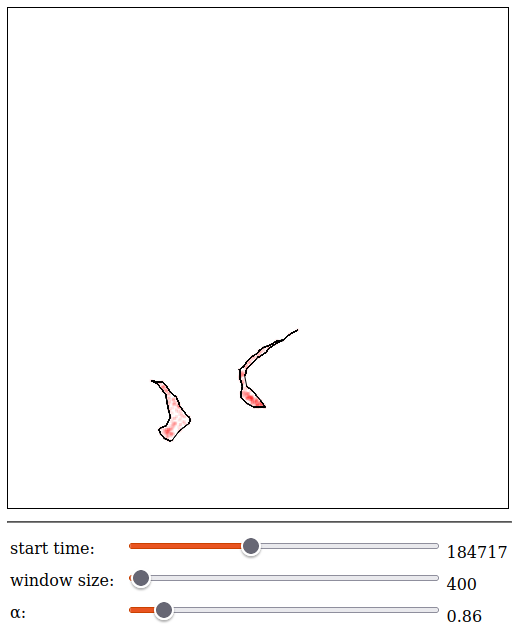}%
		\includegraphics[width=.3333\textwidth]{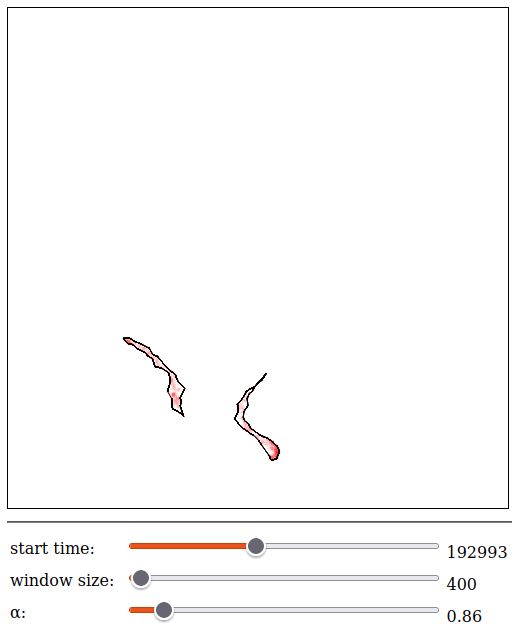}
		\includegraphics[width=.3333\textwidth]{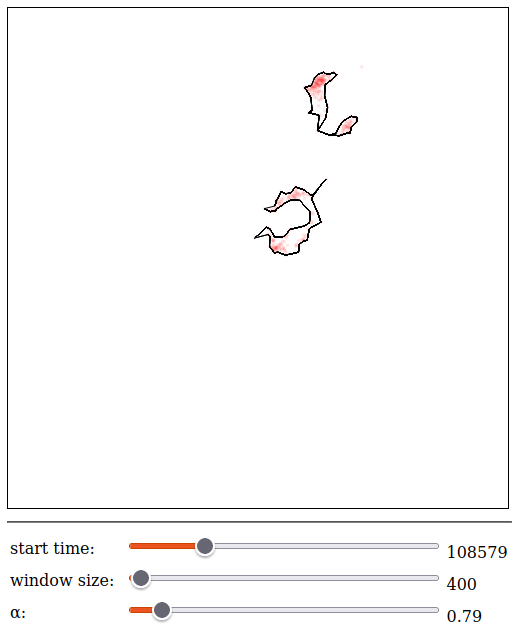}%
		\includegraphics[width=.3333\textwidth]{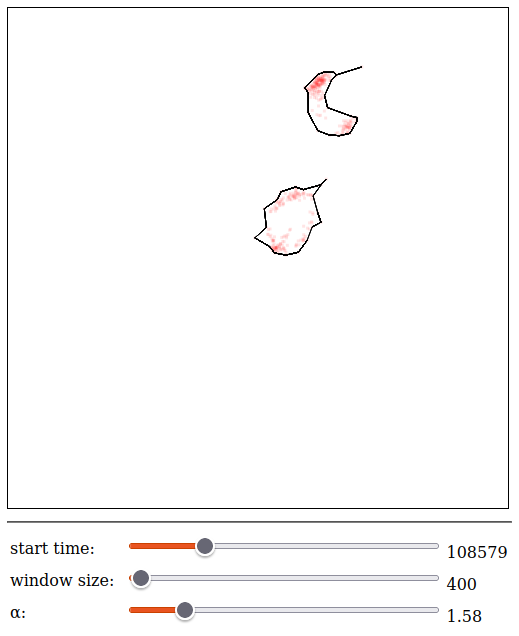}%
		\includegraphics[width=.3333\textwidth]{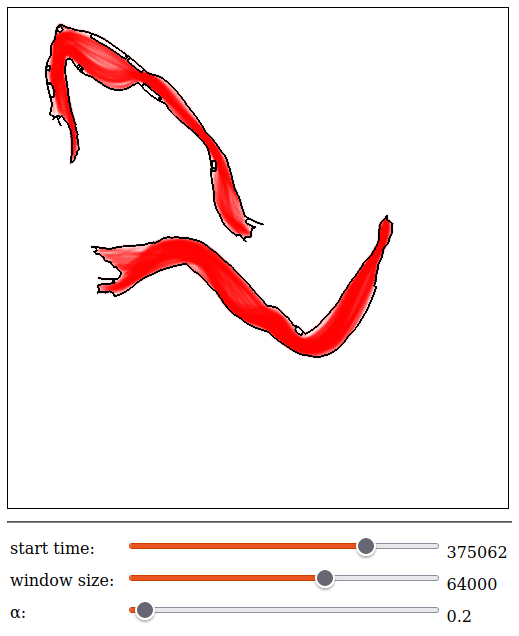}
		\caption{Interactive visualization of spatio-temporal point sets using $\alpha$-shapes. The box has a side length of 40 units. The particle sets in red are overlaid only for reference. \textbf{Top:} Following particle swarm movement by advancing the time window. \textbf{Bottom left and bottom middle:} Changing the $\alpha$-value to control the level of detail in the outline. \textbf{Bottom right:} Revealing swarm trajectories using a larger time window.}\label{fig:swarmDemo}
	\end{figure*}

	\section{Demo Application}\label{sec:demo}
		We create an interactive visualization of the temporal $\alpha$-shape $\alpha_T$ as a demo application. In order to query $\alpha_T$ we need to be able to identify all cuboids which contain the point $(i,j,\alpha)$ for a given time window $P_{i,j}$ and $\alpha$-value. A rectangle stabbing query data structure does exactly that, so we use a \emph{cs-priority-box-tree}~\cite{Agarwal2002} to identify the cuboids which match a given query. Our user interface provides three sliders to adjust the time window and value of $\alpha$ in real time. As sliders are moved, the cs-priority-box-tree is queried to fetch the corresponding cuboids, and the corresponding edges are rendered. Figure~\ref{fig:swarmDemo} shows how this application helps visualize spatio-temporal point sets by inspecting different lengths of time windows, and by adjusting the value of $\alpha$. A video demo is available online at \url{https://youtu.be/Esh7_uzmBac}.

		The number of displayed edges is typically orders of magnitude smaller than the point set they represent. Unlike this small example, real-world applications can have too many points to handle user interaction in real time. Then only an outline-based representation is feasible, making the ability to efficiently pre-compute visual representations even more important.

\section{Experimental Results}\label{sec:results}
		We evaluate a 2D Java implementation of our algorithm to compute both $T$ and $\alpha_T$. An implementation which works in arbitrary dimension is available on GitHub~\cite{weitbrecht2023github}. We store coordinates with double precision and we use objects for edges and triangles. We use the following systems for our experiments:
		\begin{itemize}
			\item \textit{system1.} Intel\textregistered{} Core i5-10500 CPU @ 3.10GHz and 64 GB of RAM.
			\item \textit{system2.} Intel\textregistered{} Xeon\textregistered{} E5-2650v4 CPU @ 2.20GHz and 768 GB of RAM.
		\end{itemize}
		We use the following data sets in our experiments:
		\begin{itemize}
			\item \textit{storm.} Storm events observed in the United States in the years 1991 to 2000, provided by the NOAA~\cite{stormDataSet}. Data points with identical timestamps or coordinates were removed, resulting in 60,173 data points in total. This data set was also used to evaluate the approach in~\cite{Haunert19}, where slightly different duplicate filtering leads to a version of this data set with 59,789 points.
			\item \textit{swarm.} A simulated 2D particle animation of 400 particles, inspired by the \emph{swarm} screensaver (\url{https://youtu.be/epskMJVIRXY}). There are 1,200 movement steps in which 398 follower particles always follow the closest of two leader particles, see Figure~\ref{fig:swarmDemo}. A copy of each particle is created for each movement step, so each of the 1,200 movement steps is encoded with 400 timestamps and we get 480,000 data points in total. For numerical stability, we perturb points randomly in a radius of $\mathcal{O}(10^{-7})$ times the canvas size.
		\end{itemize}
		
		\subsection{Naive Approach}
			To motivate precomputing $\alpha_T$ we benchmark the naive approach. We have implemented a randomized incremental construction of the Delaunay triangulation in the same codebase as the temporal $\alpha$-shape code. Point location is accomplished by tracing the location of points through the history of triangles which have contained them during the incremental construction so far~\cite{Guibas1992}.
			
			We compute the Delaunay triangulation of time windows of various sizes in the \textit{swarm} data set and identify the $\alpha$-edges for various $\alpha$-values. Table~\ref{tab:ricTimings} gives timings on \textit{system1} and the number of computed Delaunay triangles and resulting $\alpha$-edges. It shows that computing the entire Delaunay triangulation is quite wasteful when only the $\alpha$-shape is required, and that querying a precomputed data structure provides a speedup of multiple orders of magnitude. For example, for a time window size of $2^{18}$ and $\alpha=0.8$ we spend around 1.6 ms per $\alpha$-edge with the naive approach whereas the query data structure only takes 1.7 \textmu s per $\alpha$-edge.

			\begin{table*}[h]
					\caption{Comparing the performance of computing $\alpha$-shapes of requested time windows in the \textit{swarm} data set naively to that of querying a query data structure. We give for various $\alpha$-values and time window sizes $len$ the time to compute the Delaunay triangulation and its size, and the number of $\alpha$-edges for the requested time window and $\alpha$-value, and the time to query these $\alpha$-edges from a precomputed query data structure (see also Section~\ref{subsec:resultsQuery}) which maintains the temporal $\alpha$-shape of the whole point set (480,000 points). Query times are averaged over 10 runs. We include convex hull facets in the Delaunay triangulation, so the size is exactly $2 \cdot len - 2$.} \label{tab:ricTimings}
					\centering
					\tabcolsep=0.058cm
					\begin{tabular}{c c c || c c c c || r r r r r r}
						\multicolumn{3}{c||}{\textbf{Delaunay triangulation}} & \multicolumn{4}{c||}{\textbf{$\boldsymbol{\alpha}$-edges in $\boldsymbol{\alpha}$-shape}} & \multicolumn{6}{c}{\textbf{Query time after preprocessing}}\\
						$len$ & time & triangles & $\alpha=0.2$ & $\alpha=0.8$ & $\alpha=2.0$ & $\alpha=5.0$ & ~ & $\alpha=0.2$ & $\alpha=0.8$ & $\alpha=2.0$ & $\alpha=5.0$ \\
						\hline \rule{0pt}{0.9\normalbaselineskip}%
						$2^{12}$ & 16ms & 8,190 & 3,632 & 250 & 130 & 93 & ~ & 1,417\textmu s & 573\textmu s & 384\textmu s & 295\textmu s\\
						$2^{14}$ & 76ms & 32,766 & 3,187 & 291 & 190 & 129 & ~ & 1,315\textmu s & 728\textmu s & 564\textmu s & 482\textmu s \\
						$2^{16}$ & 370ms & 131,070 & 3,092 & 500 & 328 & 176 & ~ & 1,073\textmu s & 741\textmu s & 646\textmu s & 559\textmu s \\
						$2^{18}$ & 2,068ms & 524,286 & 3,962 & 1,292 & 752 & 392 & ~ & 2,593\textmu s & 2,193\textmu s & 2,131\textmu s & 1,420\textmu s \\
					\end{tabular}
			\end{table*}

		\subsection{Temporal \texorpdfstring{$\boldsymbol{\alpha}$}{α}-Shape Computation}
			We benchmark the computation of the temporal $\alpha$-shape on our data sets, split into the Delaunay triangulation phase and the $\alpha$-shape phase. Table~\ref{tab:alphaComputation} shows our results. Runtime measurements are the average of three runs.
			\begin{table*}[h]
					\caption{Computation time, memory usage and output size of the temporal Delaunay enumeration and $\alpha$-shape.} \label{tab:alphaComputation}
					\centering
					\begin{tabular}{c | c c c || c c c }
						~ & \multicolumn{3}{c||}{\textbf{Temporal Delaunay enumeration}} & \multicolumn{3}{c}{\textbf{Temporal $\boldsymbol{\alpha}$-shape}}\\
						data set/system & time & $|T|$ & memory & time & $|\alpha_T|$ & memory \\
						\hline
						\textit{storm/system1} & 10.8s & 6,460,098 & 7.9 GB & 9.9s & 18,231,685 & 3.1 GB \\
						\textit{swarm/system2} & 147.3s & 45,494,688 & 90.4 GB & 222.1s & 127,227,279 & 28.4 GB \\
					\end{tabular}
			\end{table*}

			We compare our approach to compute the temporal $\alpha$-shape against that of~\cite{Haunert19}, which was benchmarked using the \textit{storm} data set on an Intel\textregistered{} Core i7-7700T CPU. Measurements from~\cite{Haunert19} are not given in detail, so we have made the following assumptions in order to compare the performance of the two approaches: Construction time of the query data structure in~\cite{Haunert19} is not given separately from the $\alpha$-edge computation time (\cite{Haunert19}, p. 10). Using the total construction time (``less than two minutes'') for a smaller $\alpha$-value, which produces a larger query data structure than for the larger $\alpha$-value, as an upper bound on the query data structure construction time, we subtract ``less than two minutes'' from ``about 20 minutes'' to obtain a lower bound on the $\alpha$-edge computation time. The query time of 120 ms (\cite{Haunert19},~p. 9) is only given for $\alpha=200$, though it is unlikely that the query time would decrease significantly for $\alpha=500$ given that the resulting query data structure has similar size.

			Table~\ref{tab:previousCompare} compares the construction and query performance of~\cite{Haunert19} to that of our approach for the \textit{storm} data set. Our approach gives a speedup in construction time of at least $\sim$52 over~\cite{Haunert19}, and likely much more for larger $\alpha$-values, all while considering not just a single value of $\alpha$, but all values at once. Query performance also is multiple orders of magnitude better.
			
			\begin{table}[h]
					\caption{Comparison against existing approach~\cite{Haunert19} which requires a fixed $\alpha$-value (here: $\alpha = 500$) in advance.} \label{tab:previousCompare}
					\centering
					\begin{tabular}{c | c c }
						~ & \cite{Haunert19} & This paper \\
						\hline
						Finding $\alpha$-edges & about 18 minutes &  10.8s + 9.9s \\
						Query time & $\sim$120ms & $<$250\textmu s  \\
					\end{tabular}
			\end{table}

		\subsection{Query Data Structure}\label{subsec:resultsQuery}
			We benchmark our implementation of the cs-box-tree behind the demo application in Section~\ref{sec:demo}, which is implemented in C++ to allow for better memory management, with the \textit{swarm} data set. It takes 254 seconds to set up the cs-box-tree with the 127 million cuboids of the temporal $\alpha$-shape on \textit{system1}, which requires 12.8 GB of memory.
	
			Table~\ref{tab:ricTimings}, right, shows how query time tends to grow with increasing query result set size and with increasing time window size. We execute 10,000 queries with randomly selected lower and upper time window endpoints for $\alpha \in \{0.2, 0.8, 2.0, 5.0\}$ and report statistics on query result set size and query times in Table~\ref{tab:csboxtreeStats}. Query times are in the range of 0 - 3~ms, even for small $\alpha$-values. Most queries take 0.4 - 2 \textmu s per cuboid. It is clear that a precomputed query data structure provides significant speedups over computing $\alpha$-shapes on demand, especially when the resulting $\alpha$-shape representation is small.
			
			\begin{table*}[h]
					\caption{Average, min, max and standard deviation of query times and query result set sizes of 10,000 queries.} \label{tab:csboxtreeStats}
					\centering
					\begin{tabular}{c | c c c c || r r r r}
						~ & \multicolumn{4}{c||}{\textbf{Query result set size}} & \multicolumn{4}{c}{\textbf{Query time}}\\
						~ & $\alpha=0.2$ & $\alpha=0.8$ & $\alpha=2.0$ & $\alpha=5.0$ & $\alpha=0.2$ & $\alpha=0.8$ & $\alpha=2.0$ & $\alpha=5.0$ \\
						\hline \rule{0pt}{0.9\normalbaselineskip}%
						min & 2 & 10 & 6 & 6 & 34\textmu s & 49\textmu s & 56\textmu s & 39\textmu s\\
						avg & 1,847 & 846 & 570 & 289 & 778\textmu s & 664\textmu s & 577\textmu s & 473\textmu s \\
						max & 4,868 & 1,913 & 1,215 & 559 & 2,958\textmu s & 1,747\textmu s & 1,746\textmu s & 1,353\textmu s \\
						$\sigma$ & 1,170 & 506 & 321 & 143 & 380\textmu s & 282\textmu s & 235\textmu s & 186\textmu s \\
					\end{tabular}
			\end{table*}

		\subsection{Size of the Temporal \texorpdfstring{$\boldsymbol{\alpha}$}{α}-Shape}
			Queries of very short time windows or very small $\alpha$-values may be irrelevant for real-world applications. We investigate how restricting queries by enforcing a minimum time window size and a minimum value of $\alpha$ affects the number of cuboids necessary to correctly answer queries. Note that restricting the maximum window size and maximum value of $\alpha$ would have little effect on the number of necessary cuboids. This is because the cuboids this would exclude come from Delaunay triangles whose circumcircle is either very large, or empty in many time windows, or both. As the number of points increases, this becomes increasingly unlikely.
			
			The \textit{swarm} data set encodes each movement step of the 400 particles with 400 distinct timestamps. Let us only consider time windows aligned with movement steps, i.e.\ time windows of the type $P_{400*i, 400*j-1}$. We restrict queries to longer time windows by requiring $j-i \geq minLen$ and count how many cuboids are necessary for various values of $minLen$ in Table~\ref{tab:restrictedCuboids}, left. Table~\ref{tab:restrictedCuboids}, right, shows how many cuboids we need if we additionally enforce a minimum value of 0.1 for $\alpha$. For reference, the canvas in Figure~\ref{fig:swarmDemo} has a side length of 40 units. This cutoff value was chosen such that lower values would start producing fragmented $\alpha$-shapes. Figure~\ref{fig:minAlpha} shows the level of detail achieved with $\alpha=0.1$. 

			If we only restrict the minimum value of $\alpha$ to 0.1, but allow arbitrary time windows, 79.4\% of cuboids are necessary to answer queries correctly. Unwanted cuboids can easily be discarded after computing the temporal $\alpha$-shape, but modifying the algorithm to avoid spending time on them in the first place is anything but trivial.

			These results indicate that the entire temporal $\alpha$-shape is not much bigger than what one might actually need for real-world applications. In our case, allowing only values of 0.1 or larger for $\alpha$ and only allowing time windows covering at least 64 movement steps still requires 15.3\% of all cuboids.

			\begin{table*}[h]
					\caption{Number and percentage of cuboids necessary to answer queries aligned to movement steps, which span 400 timestamps, and with a minimum time window size of 1, 2, 4, $\dots$, 512 movement steps within \textit{swarm}.} \label{tab:restrictedCuboids} 
					\centering
					\begin{tabular}{c|cc||cc}
						\multicolumn{1}{c|}{~} & \multicolumn{2}{c||}{\textbf{No restriction on $\boldsymbol{\alpha}$}} & \multicolumn{2}{c}{$\boldsymbol{\alpha \geq 0.1}$}\\
						minimum movement steps & cuboids & \hspace{-.14em}\% of total cuboids & cuboids & \hspace{-.14em}\% of total cuboids \\
						\hline \rule{0pt}{0.9\normalbaselineskip}%
						1 & 63,321,387 & 49.77\% & 39,341,745 & 30.92\% \\
						2 & 56,385,797 & 44.32\% & 34,390,733 & 27.03\% \\
						4 & 49,494,282 & 38.90\% & 30,238,587 & 23.77\% \\
						8 & 43,987,644 & 34.57\% & 27,242,646 & 21.41\% \\
						16 & 40,037,206 & 31.47\% & 24,826,379 & 19.51\% \\
						32 & 36,599,040 & 28.77\% & 22,270,628 & 17.50\% \\
						64 & 33,197,685 & 26.09\% & 19,464,186 & 15.30\% \\
						128 & 27,928,269 & 21.95\% & 15,045,754 & 11.83\% \\
						256 & 22,108,154 & 17.38\% & 11,219,232 & 8.82\% \\
						512 & 14,968,658 & 11.77\% & 6,414,891 & 5.04\% \\
					\end{tabular}
			\end{table*}

			\begin{figure}[h]
				\centering
				\includegraphics[width=.4\textwidth]{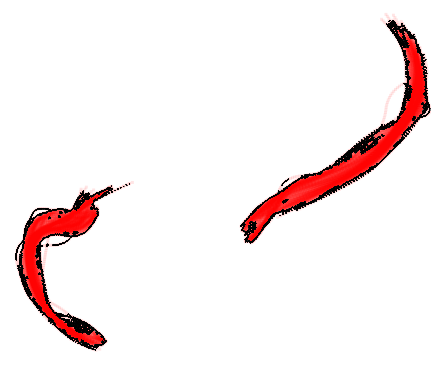}%
				\caption{Screenshot of the interactive visualization with $\alpha=0.1$, $\alpha$-shape in black.}\label{fig:minAlpha}
			\end{figure}

	\section{Conclusion}\label{sec:outlook}
		We gave an algorithm to compute the temporal $\alpha$-shape, which is a description of all $\alpha$-shapes over all time windows and over all values of $\alpha$, in output-sensitive linear time in arbitrary fixed dimension. Our approach is based on an existing framework to compute all Delaunay triangles over all time windows. A demo application verified the practicality of our approach and experimental results suggest that precomputing all $\alpha$-shapes, even those of atypical time windows and values of $\alpha$, does not cause excessive overhead in the output size. Experiments also showed that our approach achieves a speedup of at least $\sim$52 over an existing, less general approach.
		
		While our algorithm takes time linear in the size of the temporal $\alpha$-shape $\alpha_T$, we do not yet have practical bounds on $|\alpha_T|$. Further research into this direction is needed. We hope to see our work utilized for further real-world applications.

	\bibliography{DT}

@inproceedings{aganj2007spatio,
	title={Spatio-temporal shape from silhouette using four-dimensional {Delaunay} meshing},
	author={Aganj, Ehsan and Pons, Jean-Philippe and S{\'e}gonne, Florent and Keriven, Renaud},
	booktitle={2007 IEEE 11th International Conference on Computer Vision},
	pages={1--8},
	year={2007},
	organization={IEEE},
	doi={10.1109/ICCV.2007.4409016}}

@article{Agarwal2002,
	author = {Pankaj Agarwal and Mark de Berg and Joachim Gudmundsson and Mikael Hammar and Herman Haverkort},
	year = {2001},
	pages = {291-312},
	title = {Box-Trees and R-trees with Near-Optimal Query Time},
	volume = {28},
	journal = {Discrete \& Computational Geometry},
	doi = {10.1007/s00454-002-2817-1}
}

@inproceedings{Haunert19,
	author    = {Annika Bonerath and
	Benjamin Niedermann and
	Jan{-}Henrik Haunert},
	title     = {Retrieving {\(\alpha\)}-Shapes and Schematic Polygonal Approximations
	for Sets of Points within Queried Temporal Ranges},
	booktitle = {Proceedings of the 27th ACM SIGSPATIAL International Conference on Advances in Geographic Information Systems},
	pages     = {249--258},
	publisher = {{ACM}},
	year      = {2019},
	doi       = {10.1145/3347146.3359087}
}

@article{EKS83,
	author    = {Herbert Edelsbrunner and
	David G. Kirkpatrick and
	Raimund Seidel},
	title     = {On the shape of a set of points in the plane},
	journal   = {IEEE Transactions on Information Theory},
	volume    = {29},
	number    = {4},
	pages     = {551--559},
	year      = {1983},
	doi       = {10.1109/TIT.1983.1056714},
}

@inproceedings{funke2020efficiently,
	title={Efficiently Computing All {Delaunay} Triangles Occurring over All Contiguous Subsequences},
	author={Funke, Stefan and Weitbrecht, Felix},
	booktitle={31st International Symposium on Algorithms and Computation (ISAAC 2020)},
	year={2020},
	organization={Schloss Dagstuhl-Leibniz-Zentrum f{\"u}r Informatik},
	doi = {10.4230/LIPIcs.ISAAC.2020.28}
}

@Article{Guibas1992,
	author="Guibas, Leonidas J.
	and Knuth, Donald E.
	and Sharir, Micha",
	title="Randomized incremental construction of {Delaunay} and {Voronoi} diagrams",
	journal="Algorithmica",
	year="1992",
	day="01",
	volume="7",
	number="1",
	pages="381--413",
	issn="1432-0541",
	doi       = {10.1007/BF01758770},
}

@inproceedings{sussmuth2008reconstructing,
	title={Reconstructing animated meshes from time-varying point clouds},
	author={S{\"u}{\ss}muth, Jochen and Winter, Marco and Greiner, G{\"u}nther},
	booktitle={Computer Graphics Forum},
	volume={27},
	number={5},
	pages={1469--1476},
	year={2008},
	organization={Wiley Online Library},
	doi = {10.1111/j.1467-8659.2008.01287.x}
}

@misc{weitbrecht2023github,
	author = {Felix Weitbrecht},
	title = {{DelaunayEnumerator, a GitHub repository}},
	year = {2023},
	howpublished = {\url{https://github.com/felixweitbrecht/DelaunayEnumerator}}
}

@article{vauhkonen2009identification,
  title={Identification of Scandinavian commercial species of individual trees from airborne laser scanning data using alpha shape metrics},
  author={Vauhkonen, Jari and Tokola, Timo and Packal{\'e}n, Petteri and Maltamo, Matti},
  journal={Forest Science},
  volume={55},
  number={1},
  pages={37--47},
  year={2009},
  publisher={Oxford University Press},
  doi={10.1093/forestscience/55.1.37}
}

@incollection{edelsbrunner2011alpha,
  title={Alpha shapes-a survey},
  author={Edelsbrunner, Herbert},
  booktitle={Tessellations in the Sciences: Virtues, Techniques and Applications of Geometric Tilings},
  year={2011}
}

@article{edelsbrunner1994three,
  title={Three-dimensional alpha shapes},
  author={Edelsbrunner, Herbert and M{\"u}cke, Ernst P},
  journal={ACM Transactions on Graphics (TOG)},
  volume={13},
  number={1},
  pages={43--72},
  year={1994},
  publisher={ACM New York, NY, USA},
  doi={10.1145/174462.156635}
}

@inproceedings{weitbrecht2022linear,
	title={Linear Time Point Location in {Delaunay} Simplex Enumeration over all Contiguous Subsequences},
	author={Weitbrecht, Felix},
	booktitle={EuroCG},
	pages={399--404},
	year={2022},
	url={http://eurocg2022.unipg.it/booklet/EuroCG2022-Booklet.pdf}
}

@ARTICLE{8440823,
  author={Buchmüller, Juri and Jäckle, Dominik and Cakmak, Eren and Brandes, Ulrik and Keim, Daniel A.},
  journal={IEEE Transactions on Visualization and Computer Graphics}, 
  title={MotionRugs: Visualizing Collective Trends in Space and Time}, 
  year={2019},
  volume={25},
  number={1},
  pages={76-86},
  doi={10.1109/TVCG.2018.2865049}
 }

@misc{stormDataSet,
	author = {{National Climatic Data Center, NESDIS}},
	organization = {{National Centers for Environmental Information, NESDIS, NOAA, U.S. Department of Commerce}},
	title = {Storm Events Data},
	year = {1950 -- 2022},
	url = {https://www.ncei.noaa.gov/access/metadata/landing-page/bin/iso?id=gov.noaa.ncdc:C00510}
}

@article{chan2022orthogonal,
  title={Orthogonal point location and rectangle stabbing queries in 3-d},
  author={Chan, Timothy and Nekrich, Yakov and Rahul, Saladi and Tsakalidis, Konstantinos},
  journal={Journal of Computational Geometry},
  volume={13},
  number={1},
  pages={399--428},
  year={2022},
  doi={10.20382/jocg.v13i1a15}
}

@inproceedings{weitbrecht2023complexityeurocg,
	title={On the number of {Delaunay} Simplices over all Time Window in any Dimension},
	author={Weitbrecht, Felix},
	booktitle={EuroCG},
	pages={39--45},
	year={2023},
	url={https://dccg.upc.edu/eurocg23/wp-content/uploads/2023/05/Booklet_EuroCG2023.pdf}
}

@phdthesis{weitbrecht2026delaunay,
  title={Time Windowed Data Structures},
  author={Weitbrecht, Felix},
  year={2026},
  school={Universit\"{a}t Stuttgart},
  doi= {10.18419/opus-18190}
}
	
\end{document}